\documentclass[11pt]{amsart}
\usepackage[a4paper]{geometry}

\usepackage{mathrsfs}
\usepackage{amsfonts}
\usepackage{amsmath} 
\usepackage{amssymb} 
\usepackage{amsthm}
\usepackage{graphicx} 
\usepackage{url}
\usepackage{color}
\usepackage{hyperref}
\definecolor{refcolor}{RGB}{0,0,190}
\hypersetup{
    colorlinks,
    citecolor=refcolor,
    filecolor=refcolor,
    linkcolor=refcolor,
    urlcolor=refcolor
}

\theoremstyle{definition}
\newtheorem{theorem}{Theorem}

\newtheorem{remark}{Remark}

\newtheorem{example}{Example}

\newtheorem{hypothesis}{Hypothesis}
\newtheorem{property}{Property}

\def\({\left(}
\def\){\right)}

\newcommand{\hilbert}{\mathscr{H}}

\newcommand{\mc}[1]{\mathcal{#1}}

\newcommand{\qmU}{$\mathscr{U}$}
\newcommand{\qmR}{$\mathscr{R}$}

\newcommand{\C}{\mathbb{C}}

\newcommand{\tr}{\textnormal{tr}}

\newcommand{\tn}{\textnormal}

\newcommand{\dsfrac}[2]{\displaystyle{\frac{#1}{#2}}}

\newcommand{\ie}{\textit{i.e.} }

\newcommand{\eg}{\textit{e.g.} }

\newcommand{\schrod}{Schr\"odinger}

\newcommand{\bra}[1]{\langle#1|}
\newcommand{\ket}[1]{|#1\rangle}
\newcommand{\braket}[2]{\langle#1|#2\rangle}

\def\sref #1{\S\ref{#1}}

\newcommand{\image}[3]{
\begin{figure}[!ht]
\includegraphics[width=#2\textwidth]{#1}
\caption{\small{\label{#1}#3}}
\end{figure}
}

\title{Quantum Measurement and Initial Conditions}

\author{Ovidiu Cristinel Stoica*}
\thanks{*Department of Theoretical Physics, National Institute of Physics and Nuclear Engineering -- Horia Hulubei, Bucharest, Romania. Email: \href{mailto:cristi.stoica@theory.nipne.ro}{cristi.stoica@theory.nipne.ro},  \href{mailto:holotronix@gmail.com}{holotronix@gmail.com}}
\date{\today}

\begin{document}

\begin{abstract}
Quantum measurement finds the observed system in a collapsed state, rather than in the state predicted by the Schr\"odinger equation.
Yet there is a relatively spread opinion that the wavefunction collapse can be explained by unitary evolution (for instance in the decoherence approach, if we take into account the environment).

In this article it is proven a mathematical result which severely restricts the initial conditions for which measurements have definite outcomes, if pure unitary evolution is assumed. This no-go theorem remains true even if we take the environment into account.
The result does not forbid a unitary description of the measurement process, it only shows that such a description is possible only for very restricted initial conditions.

The existence of such restrictions of the initial conditions can be understood in the four-dimensional block universe perspective, as a requirement of global self-consistency of the solutions of the Schr\"odinger equation.
\end{abstract}

\maketitle

\tableofcontents


\section{Introduction}

\subsection{Initial conditions in quantum mechanics}

Quantum mechanics is usually presented as consisting of two processes (von Neumann, \cite{vonNeumann1955foundations}). The first one is the \textit{unitary evolution}, or \textit{the {\qmU} process}
\begin{equation}
\label{eq_unitary}
\ket{\psi(t)}=U(t,t_0)\ket{\psi_0},
\end{equation}
obtained by solving the {\schrod} equation with initial condition $\ket{\psi(t_0)}=\ket{\psi_0}\in\hilbert$, where $\hilbert$ is a \textit{Hilbert space} (fig. \ref{qm_u_evolution.pdf}).

\image{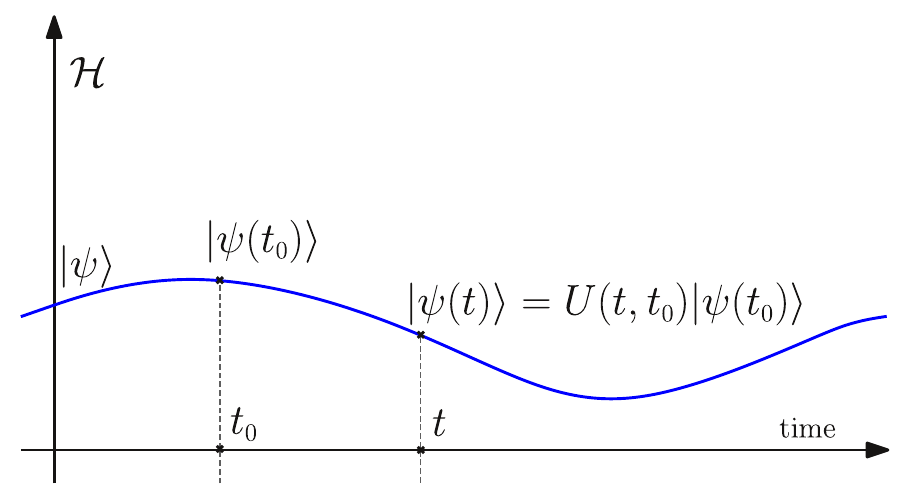}{0.6}{The unitary evolution, or the {\qmU} process.}

A classical system is determined by a set of partial differential equations, and initial conditions. Initial conditions are determined by an experiment performed at the time $t_0$ (within an error inherent to measurements).
In classical mechanics, the observation process can find the system in any allowed state.

By \emph{measurement} we will understand a quantum measurement, performed with a \emph{measurement device} or \emph{apparatus}, as described by von Neumann \cite{vonNeumann1955foundations}. Accordingly, what we measure are \emph{observables}, which are Hermitian operators defined on the Hilbert space of the observed system. The outcome of the measurement is an \emph{eigenvalue} of the observable, and the observed system is found to be in an \emph{eigenstate} of the observable.

The measurement of the quantum state of the system is considered to trigger the second process, the \textit{wavefunction collapse}, or the \textit{state vector reduction} process {\qmR} (fig. \ref{qm_reduction_0.pdf}).
This consists of projecting the state of the quantum system on an eigenstate of the measured observable, resetting by this its initial conditions.

\image{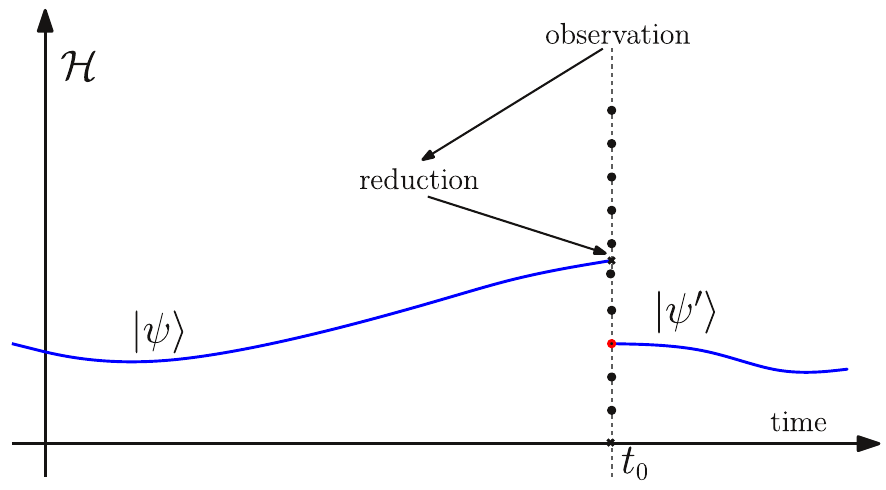}{0.6}{The wavefunction collapse, or the {\qmR} process.}

\subsection{The density matrix formalism}

More generally, we can consider instead of the state vector $\ket{\psi}\in\hilbert$, a \textit{density operator} (or matrix) $\rho$, which is Hermitian on $\hilbert$. In at least one orthonormal basis $\(\ket{\psi_i}\)_i$, the density matrix $\rho$ has the diagonal form
\begin{equation}
\label{eq_density_matrix}
\rho=\sum_ip_i\ket{\psi_i}\bra{\psi_i},
\end{equation}
where $p_i\geq 0$, $\sum_ip_i=1$. The density matrix can be interpreted as a \textit{statistical ensemble} (``improper mixture'', by the terminology of d'Espagnat \cite{d'E76}), where $p_i$ is the probability that the system is in the state $\ket{\psi_i}$. It can also be understood as a \textit{reduced density matrix} of a pure state from a higher dimensional Hilbert state (``proper mixture''), representing all the information contained in a system which is entangled with another system which is ignored. 

The {\qmU} process for a density matrix $\rho$ is described with the help of the time evolution operator, by
\begin{equation}
\label{eq_density_matrix_unitary}
\rho(t)=U(t,t_0)\rho(t_0)U(t,t_0)^{-1}.
\end{equation}
In the decoherence interpretation \cite{joos-zeh1985emergence,Zeh96,Zur98}, to apply the {\qmR} process, the density matrix should be \textit{decohered}, that is, it should be diagonal in an eigenbasis of the observable. Then, it is interpreted as a statistical ensemble, and the probability to find the system in the state $\ket{\psi_i}$ is given by $p_i$.

\subsection{Is unitary evolution violated during measurement?}

Unitary evolution seems to be ubiquitous. The exception, and the reason for the introduction of the {\qmR} process, is that when the system is in a certain state, a subsequent measurement may project it to a different state. But one does not exclude the possibility that, when we consider in addition to the observed system, the environment (including the measurement apparatus and the observer) and the interactions between these systems, the evolution turns out to be unitary.

The viewpoint that unitary evolution is enough, and can account for the {\qmR} process too, gained more and more supporters lately, due to the development of the \emph{many worlds interpretation} (MWI) \cite{Eve57,Eve73,dW71,dWEG73}, the \textit{consistent histories} interpretation \cite{Gri84,Omn88,GH90a}, and especially of the \textit{decoherence program} \cite{joos-zeh1985emergence,Zeh96,Zur98}.

A pole took place at a conference on quantum computation, at the Isaac Newton Institute in Cambridge, in July 1999. The question \textit{``Do you believe that all isolated systems obey the {\schrod} equation (evolve unitarily)?''} received 59 answers of ``yes'', 6 of ``no'', and the remaining 31 physicists were undecided.
Tegmark and Wheeler commented about this \cite{tegmark100quantum}
\begin{quote}
although these [quantum textbooks] infallibly list explicit non-unitary collapse as a fundamental postulate in one of the early chapters, the poll indicates that many physicists -- at least in the burgeoning field of quantum computation -- no longer take this seriously.
\end{quote}

The decoherence approach is based on the idea that the interactions with the environment cause the density matrix of the observed system to become diagonal, in a preferred basis. After diagonalization, which is viewed as a \textit{pre-measurement}, the density matrix can be interpreted as a statistical ensemble. Presumably, once two branches decohered, they no longer interact with one another, this leading to an ``effective collapse'', which would replace the {\qmR} process.

Other proposal that for quantum mechanics the {\qmU} process is enough, even for the apparent {\qmR} process, was made by the author in \cite{Sto08b,Sto12QMa}, and another one by 't Hooft in \cite{hooft2011wave,tHooft2014CellularAutomatonInterpretationQM}.

Regarding the universality of the {\qmU} process, the opinions are divided. More details, and deep analyses of these opinions, are listed and discussed in \cite{adler2003decoherence,schlosshauer2005decoherence,sch07}. Therefore, it is justified to consider the hypothesis:
\begin{hypothesis}
\label{hypo:unitary}
The {\qmR} process is reducible to the {\qmU} process.
\end{hypothesis}

We will show that from this hypothesis follows that only a small part of the possible initial states of the observed system can evolve into eigenstates of the observable.

\subsection{Standard arguments in favor of a discontinuous collapse}

The introduction by von Neumann \cite{vonNeumann1955foundations} of the state vector reduction process {\qmR} is justified by the argument that the unitary evolution process {\qmU} would transform the system $\ket{\eta}\ket{\psi}$ made of the observed system $\ket{\psi}$ and the apparatus $\ket{\eta}$ into a superposition of the form $\sum_{i\in\sigma(\mc O)}\ket{\eta}_i\ket{\psi}_i$, where $\sigma(\mc O)$ is the spectrum of the observable $\mc O$, and $\ket{\psi}_i$ are eigenstates, while in reality we never get such a superposition, but only one term from the sum, corresponding to only one eigenvalue, according to the Born rule.

The standard answer to this argument given by the supporters of Hypothesis \ref{hypo:unitary} is based on decoherence: if we introduce enough additional extra variables, in the form of the environment, the density matrix of the observed system becomes diagonal in the eigenbasis of the observable. Then, since the density matrix is diagonal, we can just interpret it as a statistical ensemble, and treat the probability as being classical, resulting from our lack of knowledge of the initial conditions. A weakness of this kind of argument is that if we choose a different observable instead of $\mc O$, one which does not commute with $\mc O$, we obtain a different decomposition of the density matrix as a mixture, so the initial conditions admitted by the new observable will be different. This already shows that the initial conditions of the observed system have to depend on those of the apparatus, if we want to assume unitary evolution. In this article, we will give a mathematical proof that this happens in general, no matter how we invoke the environment.
Theorem \ref{thm_unitary_ic} will show that if we assume unitary evolution during the measurement process, the condition that the initial conditions are very special is unavoidable.

\section{The property of restricted initial conditions}

Let $\mc S$ be the set of all mathematically possible states of the observed system. For example, the states can be the rays or density matrices in a Hilbert space $\hilbert$, but they can be any kinds of quantities which are supposed to reproduce the predictions of quantum mechanics.
Such a system is said to have {\em the property of restricted initial conditions} if it satisfies the following:

\begin{property}[of restricted initial conditions]
\label{property_ic}
Not all mathematically possible initial states of a system lead to physically acceptable states.
\end{property}

In this article we are interested in a specific type of ``physically acceptable states'', which are those representing definite results of quantum measurements. We will see that, in order to have definite outcomes of measurements by unitary evolution, the initial conditions have to be restricted.

\begin{example}
\label{example_stationary}
For example, consider two quantum systems represented by the state vectors $\ket{\mu}\in\hilbert_{\mu}$ and $\ket{\psi}\in\hilbert_{\psi}$.
Suppose they satisfy the following conditions:
\begin{enumerate}
	\item 
$\ket{\psi}$ represents the observed system,
	\item 
$\ket{\psi}$ is stationary, {\ie} does not change in time (except by a phase factor),
	\item 
$\ket{\mu}$ represents an apparatus measuring an observable $\mc O$ (a Hermitian operator on the Hilbert space of the observed system),
	\item 
the measurement does not disturb the observed system.
\end{enumerate}
If these conditions are satisfied, then the only possible initial conditions of $\ket{\psi}$ are those which are eigenstates of the observable $\mc O$.
\end{example}

\begin{example}
\label{example_undisturbed}
Let us drop the condition 2 from the previous example, and assume that the observed system evolves unitarily, but the measurement still does not disturb it.
The unitary evolution is reversible, in the sense that by knowing the state $\ket{\psi(t)}$ at a future time $t$, we can determine the state $\ket{\psi(t_0)}$ at the initial time $t_0$, by $\ket{\psi(t_0)}=U_{\psi}^{-1}(t,t_0)\ket{\psi(t)}$. This allows us to determine the initial condition that led to the observed eigenstate. 
This shows that there is an initial state of the observed system which became, by unitary evolution alone, the observed eigenstate. So in this case there is a description of the measurement process by unitary evolution.

The price to be paid is again that relying solely on the {\qmU} process requires very special initial conditions. The initial conditions had to be from the very \textit{beginning} in such a way so that \textit{later}, when the measurement is performed, an eigenstate of the observable is obtained (fig. \ref{qm_outcomes_2_init_cond.pdf}). Any other initial conditions are unphysical, in the sense that they do not lead to definite outcomes of the measurements.

\image{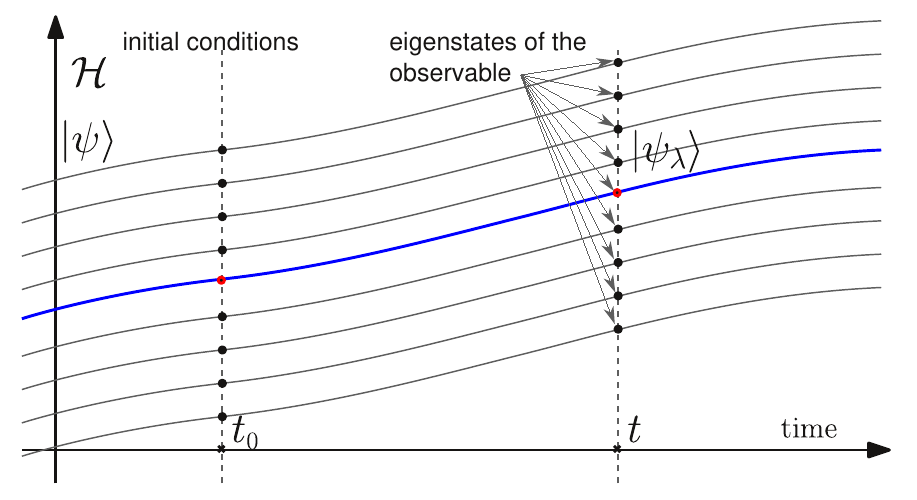}{0.6}{A measurement performed at the time $t$ finds the observed system in an eigenstate of the observable $\mc O$. Assuming that the measurement did not disturb the observed system, this can be explained by unitary evolution only if we admit that the initial state was already at $t_0$ an eigenstate of the observable $\mc O'=U^{\dagger}(t,t_0)\mc O U(t,t_0)$. This means that the initial conditions had to be very special, in order to obtain definite outcomes of the measurement.}

\end{example}

The previous examples show that, under those assumptions, the initial state of the observed system has to depend on the state of the measurement device (hence on its initial state). But one can still object that the assumptions made were too strong, that in reality the apparatus disturbs the observed system, if this is not already an eigenstate of the observable. One can also object that the environment interacts with the observed system too. But we will see that the additional liberty obtained by introducing interactions with the apparatus and any sort of environment cannot avoid the conclusion that, if the measurement takes place without breaking the unitary evolution, the system has Property \ref{property_ic}.


\section{If the collapse is assumed to be unitary}
\label{s_uqm}

\subsection{Unitary evolution, measurement, and initial conditions}
\label{s_measurement_ic}

Can measurements make a quantum system become an eigenvalue of the observable, for any initial state of the observed system, just by unitary evolution? More precisely, let $\ket{\psi}\in\hilbert_{\psi}$ be a quantum system, and $\mc{O}$ an observable corresponding to the system. Assume that the observable is measured by a system $\ket{\mu}\in\hilbert_{\mu}$, which is considered to be a quantum system. We can consider the measurement apparatus $\ket{\mu}$ as containing the environment, in the sense of the decoherence program. It is often claimed that this ingredient can help the observed system to become an eigenstate of the observable. We are interested if it is possible that the following conditions are simultaneously satisfied:
\begin{enumerate}
	\item 
Initially, the observed system $\ket{\psi}$ and the measurement apparatus $\ket{\mu}$ are considered to be separated. The measurement apparatus should have no prior ``knowledge'' about the observed system, and should not be entangled with it before the measurement. Hence, the total initial state is $\ket{\mu}\ket{\psi}$.
	\item 
The measurement performed by $\ket{\mu}$ finds the observed system to be an eigenstate of the observable $\mc{O}$, for any possible initial state $\ket{\psi}\in\hilbert_{\psi}$ (by disturbing it if necessary).
	\item 
For any eigenvalue $\lambda$ of the observable, there is an initial value of the observed system, so that the outcome of the measurement is $\lambda$.
	\item 
This is achieved by unitary evolution only.
	\item 
Since the observed system is found to be in an eigenstate $\ket{\psi'}$ of the observable, its state is pure, hence is separate from that of the apparatus, which is therefore pure too, say $\ket{\mu'}$. Hence, the total state after the measurement is of the form $\ket{\mu'}\ket{\psi'}$.
\end{enumerate}

The following theorem shows that in general it is not possible to satisfy all these conditions.

\begin{theorem}
\label{thm_unitary_ic}
Let $\hilbert_{\mu}$ and $\hilbert_{\psi}$ be two separable Hilbert spaces,
$\mc{O}$ a Hermitian operator on $\hilbert_{\psi}$ which has at least two distinct eigenvalues, and $\ket{\mu}\in\hilbert_{\mu}$ fixed.
Let $U:\hilbert_{\mu}\otimes\hilbert_{\psi}\to \hilbert_{\mu}\otimes\hilbert_{\psi}$ be a unitary operator so that for any eigenvalue $\lambda$ of $\mc O$ there is at least a vector $\ket{\psi}\in\hilbert_{\psi}$ for which $U\(\ket{\mu}\ket{\psi}\)$ has the form
\begin{equation}
\label{eq:unitary_ic}
U\(\ket{\mu}\ket{\psi}\) = \ket{\mu'}\ket{\psi'},
\end{equation}
where $\ket{\psi'}$ is an eigenvector of the Hermitian operator $\mc{O}$ corresponding to $\lambda$.
For this process to count as measurement, we require that at least two eigenvectors $\ket{\psi_1'}$ and $\ket{\psi_2'}$ corresponding to distinct eigenvalues of $\mc{O}$ are obtained in the right hand side of equation \eqref{eq:unitary_ic}.
Then, there are vectors $\ket{\psi}\in\hilbert_{\psi}$ for which there is no eigenvector $\ket{\psi'}$ of $\mc{O}$ satisfying \eqref{eq:unitary_ic}.
\end{theorem}
\begin{proof}
Let $\ket{\psi_1'}$ and  $\ket{\psi_2'}$ be two orthogonal eigenvectors of $\mc{O}$ in $\hilbert_{\psi}$, so that
\begin{equation}
U\(\ket{\mu}\ket{\psi_1}\) = \ket{\mu_1'}\ket{\psi_1'}
\end{equation}
and
\begin{equation}
U\(\ket{\mu}\ket{\psi_2}\) = \ket{\mu_2'}\ket{\psi_2'}
\end{equation}
for some vectors $\ket{\psi_1},\ket{\psi_2}\in\hilbert_{\psi}$.
For any two complex numbers $\alpha_1$ and  $\alpha_2$,
\begin{equation}
U\(\ket{\mu}\(\alpha_1\ket{\psi_1}+\alpha_2\ket{\psi_2}\)\) = \alpha_1\ket{\mu_1'}\ket{\psi_1'} + \alpha_2\ket{\mu_2'}\ket{\psi_2'}.
\end{equation}
Suppose there are $\ket{\mu''}\in\hilbert_{\mu}$ and an eigenvector $\ket{\psi''}\in\hilbert_{\psi}$ of $\mc{O}$, so that
\begin{equation}
\alpha_1\ket{\mu_1'}\ket{\psi_1'} + \alpha_2\ket{\mu_2'}\ket{\psi_2'} = \ket{\mu''}\ket{\psi''}.
\end{equation}
Because $\braket{\psi_1'}{\psi_2'}=0$, this can only happen if $\ket{\mu_2'}=\beta\ket{\mu_1'}$ for some $\beta\in\C$. But then, 
\begin{equation}
\ket{\mu_1'}\(\alpha_1\ket{\psi_1'} + \alpha_2\beta\ket{\psi_2'}\) = \ket{\mu''}\ket{\psi''}.
\end{equation}
From this it follows that for any $\alpha_1,\alpha_2\in\C$, the linear combination $\alpha_1\ket{\psi_1'} + \alpha_2\beta\ket{\psi_2'}$ is an eigenvector. This can only happen if the eigenvectors $\ket{\psi_1'}$ and $\ket{\psi_2'}$ correspond to the same eigenvalue.
But according to the hypothesis, there are at least two initial states $\ket{\psi_1},\ket{\psi_2}\in\hilbert_{\psi}$, which evolve into eigenvectors $\ket{\psi_1'}$ and $\ket{\psi_2'}$ corresponding to distinct eigenvalues of $\mc{O}$. It follows that the linear combinations of the form $\alpha_1\ket{\psi_1}+\alpha_2\ket{\psi_2}$, where $\alpha_1\neq 0$ and $\alpha_2\neq 0$, do not evolve into eigenvectors of $\mc{O}$.
This shows that from all possible initial states $\ket{\psi}$, only a subset of measure zero can satisfy equation \eqref{eq:unitary_ic}, concluding the proof.
\end{proof}

\begin{remark}
Theorem \ref{thm_unitary_ic} asserts that, under the assumption of unitary evolution during the measurement process, the outcome can be an eigenstate of the observable only for particular initial conditions of the total system. Either the initial state of the observed system has to depend on that of the apparatus, or the environment should depend on the initial state of the observed system to make it evolve precisely into an eigenstate of the observable. 
In both cases, the initial conditions of the observed system and of the rest of the universe have to be dependent.
If the system made of the measurement apparatus and the observed system evolves unitarily, no matter what environment we call to rescue, and no matter how we hope it affects the observed system, as long as evolution is unitary, the conclusion of the theorem cannot be avoided. Property \ref{property_ic}, of restricted initial conditions, is satisfied.
\end{remark}

\begin{remark}
If the initial conditions of the system are randomly chosen, the probability that the initial conditions are restricted so that a future measurement obtains a definite result is zero. This is because the initial conditions which can lead to definite outcomes form a union of subspaces of the Hilbert space, with dimension strict lower than that of the total Hilbert space. So the restriction imposed to the initial conditions is severe.
\end{remark}

\begin{remark}
One may think that we can avoid the conclusion of Theorem \ref{thm_unitary_ic} by relaxing the conditions. For example, we can consider that the observed system is entangled with its environment or other systems, which we assume no longer interact with it or with the measurement device. In this case, we can use its reduced density matrix. We can also relax the condition that after the measurement the observed system and the apparatus are separated. The only condition we have to keep is that at the time of measurement the reduced density matrix of the observed system becomes restricted to an eigenspace of the observable. Can this relaxation allow new possibilities to avoid the conclusion of Theorem \ref{thm_unitary_ic}? The following result shows that this is not the case.
\end{remark}

\begin{theorem}
\label{thm_unitary_ic_density_matrix}
Consider that initially, at $t=t_0$, the observed system and the measurement device are described by a density matrix of the form $\rho_0=\rho(t_0)$, and after the measurement at $t_1$ the unitary evolution operator $U=U(t_1,t_0)$ leads the total system into the state $U(\rho)=U\rho U^\dagger$ so that the reduced density matrix $\rho_{\psi(t_1)}=\tr_{\mu}U(\rho)$ is defined on an eigenspace of the observable (which is supposed to have at least two distinct eigenvalues which can be obtained as results of the measurement). Then, not all mathematically possible initial conditions lead to definite results of the measurement.
\end{theorem}
\begin{proof}
Let $\lambda_1\neq\lambda_2$ be two distinct eigenvalues which can be obtained for different initial conditions of the total system, $\rho(t_0)=\rho_{0,\lambda_1}$ and $\rho(t_0)=\rho_{0,\lambda_2}$. The existence of such initial conditions is ensured by the fact that the observable has at least two distinct eigenvalues, and unitary evolution is an isomorphism between the possible density matrices at $t_1$ and those at $t_0$. Since any convex combination of density matrices satisfying the {\schrod} equation also satisfy it, let us consider initial conditions of the form $\rho(t_0)=a\rho_{0,\lambda_1}+(1-a)\rho_{0,\lambda_2}$, where $a\in(0,1)$. Then, the reduced density matrix at $t=t_1$ is a convex combination of the form $\tr_{\mu}U(\rho)=a\rho_{1,\lambda_1}+(1-a)\rho_{1,\lambda_2}$, where $\rho_{1,\lambda_1}$ and $\rho_{1,\lambda_2}$ are defined on different eigenspaces. Hence, the convex combinations of initial conditions considered do not lead to definite results of the measurements.
\end{proof}

\subsection{Unitary measurement apparatus}

Theorem \ref{thm_unitary_ic} refers to any system $\ket{\mu}$ able to make the observed system $\ket{\psi}$ be an eigenstate of the observable. But normally a measurement apparatus, in addition, is required to leave unchanged the states which are already eigenstates of the observable.
This actually follows from the Born rule, which gives the probability $1$ that an observed system which is in an eigenstate of the observable remains unchanged.

Such an apparatus which works unitarily was discussed for example by Zurek in \cite{schlosshauer2011elegance}, p. 195--196, where he proved that measurement can only distinguish orthogonal states.

Theorem \ref{thm_unitary_ic} also applies to such an apparatus, but the condition that the eigenstates of the observable are left unchanged by the measurement process is even more strict about the admissible initial conditions under the assumption of pure unitary evolution, as the following simple result shows.

\begin{theorem}
\label{thm_measurement_apparatus}
In addition to the conditions of Theorem \ref{thm_unitary_ic}, suppose that for each eigenstate $\ket{\psi_i}\in\hilbert_{\psi}$ of the observable $\mc{O}$, there is a state vector $\ket{\mu_i}\in\hilbert_{\mu}$, so that $U\(\ket{\mu}\ket{\psi_i}\)=\ket{\mu_i}\ket{\psi_i}$. Then, the only initial states $\ket{\psi}$ which are compatible with the measurement (\ie become eigenstates of the observable) are those which already were eigenstates before the measurement.
\end{theorem}
\begin{proof}
Let $\ket{\psi}\in\hilbert_{\psi}$ be a state vector so that $U\(\ket{\mu}\ket{\psi}\)$ has the form $U\(\ket{\mu}\ket{\psi}\)=\ket{\mu'}\ket{\psi'}$ for some eigenstate $\ket{\psi'}$ of $\mc{O}$. Let $\(\ket{\psi_i}\)_i$ be an eigenbasis, so that for a particular $j$, $\ket{\psi_j}=\ket{\psi'}$. Then, we can write $\ket{\psi} = \sum_i \alpha_i \ket{\psi_i}$. From the hypothesis of Theorem \ref{thm_unitary_ic}, and because of linearity, 
\begin{equation}
U\(\ket{\mu}\ket{\psi}\) = \sum_i \alpha_i U\(\ket{\mu}\ket{\psi_i}\) = \sum_i \alpha_i \ket{\mu_i}\ket{\psi_i}.
\end{equation}
Hence
\begin{equation}
\sum_i \alpha_i \ket{\mu_i}\ket{\psi_i} = \ket{\mu'}\ket{\psi_j}.
\end{equation}
It follows that 
\begin{equation}
\label{eq_sum_vect_zero}
\sum_{i\neq j} \alpha_i \ket{\mu_i}\ket{\psi_i} + \(\alpha_j \ket{\mu_j} - \ket{\mu'}\)\ket{\psi_j} = 0.
\end{equation}
From the orthonormality of the eigenbasis $\(\ket{\psi_i}\)_i$, it follows that the nonvanishing terms appearing in the equation \eqref{eq_sum_vect_zero} are linearly independent vectors. Therefore, each of them has to be zero. This means that $\alpha_j \ket{\mu_j} = \ket{\mu'}$, and for any $i\neq j$, $\alpha_i=0$. It follows that the only initial state vectors $\ket{\psi}$ which are compatible with the measurement are already eigenstates of $\mc{O}$.
\end{proof}

\begin{remark}
\label{rem:successive_measurements}
One may think that from Theorem \ref{thm_measurement_apparatus} follows that there is no unitary evolution description of two consecutive measurements, if the observables do not have common eigenstates (which is a common situation). However, at least in some cases we can escape this by appealing to the environment to change the state from an eigenstate of the first observable into an eigenstate of the second observable (see section \sref{s:succesive_measurements}). Of course, this can only work for very special initial conditions. Moreover, no universal mechanism is known to do this. For instance, if we appeal to decoherence to provide the mechanism, if the time between successive measurements is smaller than the decoherence time, the unitary evolution may not be enough to accommodate both observations.
\end{remark}

\section{Discussion}

\subsection{Discussion of an argument for collapse}
\label{s:succesive_measurements}

A more elaborate argument in favor of a discontinuous collapse is based on successive incompatible measurements of a system. Consider for example spin measurements of a spin $\frac{1}{2}$ particle along the $x$ and $y$ axes. Suppose at the time $t_1$ we measure its spin along the $x$ axis, obtaining the state $\ket{\uparrow_x}$. If at $t_2>t_1$ we measure the spin along the $y$ axis, we obtain either $\ket{\uparrow_y}$ or $\ket{\downarrow_y}$, with equal probabilities of $\frac{1}{2}$, according to the Born rule. The eigenstates of the observable $S_y$ representing the spin along the $y$ axis are completely different from those of the spin along the $x$ axis. 
How can this be accommodated by unitary evolution?

The usual explanation is that the density matrix of the observed particle decoheres because of the environment, that is, it becomes at the time $t_2$ of the form
$$\rho = p\ket{\uparrow_y}\bra{\uparrow_y} + (1-p)\ket{\downarrow_y}\bra{\downarrow_y},$$
which is then interpreted as a statistical ensemble.
This means that the system is interpreted as being found in the resulting state, and not projected into that state. But if we apply the unitary evolution backwards in time from $t_2$ to $t_1$, can we obtain for the observed system the state $\ket{\uparrow_x}$ at $t_1$? Certainly not, at least not by the unitary evolution of the observed particle alone.

But maybe this can be achieved by the unitary evolution of the full system, containing the observed particle, the measurement devices, and the rest of the environment. The particle will appear to be subject of an interaction with the environment, which rotates the orientation of its spin. But so far it is not known an exact description of how this can occur in general situations, for example by appealing to decoherence, as explained in Remark \ref{rem:successive_measurements}. In \cite{Sto08b,Sto12QMa} is proposed a possible general solution, which involves entanglement between the observed system and the measurement device performing the previous measurement, immediately after the first measurement at $t_1$.

At any rate, the particle cannot evolve freely between $t_1$ and $t_2$ so that it changes its state from $\ket{\uparrow_x}$ to $\ket{\uparrow_y}$ or $\ket{\downarrow_y}$, but in the presence of interactions one cannot rule out such a possibility. 
What we can say is that it can only work for very special initial conditions.


The interaction with the environment we bring into discussion should be such that the particle evolves into the eigenstate of the observable measured at the later time $t_2$. Any change in the interaction will result in a different state, which is not an eigenstate of the observable. Hence, the environment itself has to be in a special initial state, which depends on the observable we measure, and of the state of the observed particle at the previous measurement. In other words, this can only work if the total system satisfied the property of restricted initial conditions, just like Theorem \ref{thm_unitary_ic} says.

Theorem \ref{thm_unitary_ic} does not forbid the possibility of a description of the measurement by pure unitary evolution, but it severely restricts any such description, by requiring the existence of special initial conditions.

\subsection{Measuring entangled particles}
\label{s:EPR}

While Theorems \ref{thm_unitary_ic} and \ref{thm_measurement_apparatus} are true when the initial state of the observed system is separate from that of the apparatus, they can be applied to the measurement of entangled particles too. This is the case of the EPR experiment \cite{EPR35,Bohm51,Bel64}.
Consider two spin $\frac{1}{2}$ particles, with Hilbert spaces $\hilbert_A$ and $\hilbert_B$, whose initial state is a singlet state
\begin{equation}
\ket{\psi}(t_0) = \dsfrac{1}{\sqrt{2}}\(\ket{\uparrow}_A\ket{\downarrow}_B - \ket{\downarrow}_A\ket{\uparrow}_B\),
\end{equation}
which decays at $t_0$ so that the particles become separated in space.
Suppose Alice measures the spin of the particle $A$, and Bob the spin of the particle $B$, along certain directions in space. Let $S_a$ and $S_b$ be the observables, each of them having the eigenvalues $\pm\frac 1 2$.

The two particles are considered to be entangled with each other, but not with other particles.
So, the state $\ket{\psi}\in\hilbert_A\otimes\hilbert_B$ representing both particles is separable from the rest of the universe.
The two measurements $S_a$ and $S_b$ performed by Alice and Bob are equivalent to measuring $\ket{\psi}$ with the observable $S_{ab}=S_a\otimes I + 3 I\otimes S_b$, where $I$ is the identity operator. We can build the total observable in other ways, for example as $S_a\otimes S_b$, or $S_a\otimes I + I\otimes S_b$, but these choices result in only two eigenvalues, while our choice ensures that there are four eigenvalues $\(\frac 1 2, \frac 3 2, \frac 5 2,\tn{ and }\frac 7 2\)$, and we get the same information as if we consider two distinct measurements of $S_a$ and $S_b$. The possible eigenstates of the observable $S_{ab}$ are all the tensor products of the eigenstates of the observables $S_a$ and $S_b$.

Now we have a quantum system which is not entangled with other systems before the measurement, and we can apply Theorems \ref{thm_unitary_ic} and \ref{thm_measurement_apparatus}.
After the measurement, the state of the two particles is separable, being of the form $\ket{\psi}(t_1)=\ket{\psi_a}\ket{\psi_b}$, where $\ket{\psi_a}$ is an eigenstate of $S_a$, and $\ket{\psi_b}$ an eigenstate of $S_b$.
If the measurement is unitary and does not change the states of the observed particles, then we can apply Theorem \ref{thm_measurement_apparatus}, and conclude that the wavefunctions were separable before the measurement, so they became separable immediately after the decay at $t_0$.
Consider now that the measurement is unitary, and the environment interacts with the particles to evolve them into eigenstates of the observables. If the interaction of each particle with its environment is local, and the two laboratories are separated by long distances and they cannot interact with each other, then again the two particles had to be separable before they entered in the laboratories.

An EPR-type experiment involving weak measurements and was analyzed in \cite{aharonov2012future-past}. The analysis shows that weak measurements made after the decay and before the strong measurements are confirmed by the strong measurements performed by Alice and Bob. While the interpretation of these conclusions was made in terms of the two-state vector formalism \cite{aharonov1964time,aharonov2007TSV}, they are also compatible with the possibility described here.

\subsection{Connections with other results}

Any unitary account of the quantum measurement process leads inevitably to the idea that initial conditions have to be very special.
But the dependence of the initial conditions of an observed particle on the future experimental setup, as in the Examples \ref{example_stationary} and \ref{example_undisturbed}, looks like retrocausality.

A similar behavior is present in Wheeler's \textit{delayed choice experiment} \cite{Whe78}. The observer can delay the choice of the observable, until the photon passes by the first beam splitter of the Mach-Zehnder interferometer (fig. \ref{fig:mach_zehnder}). Wheeler showed that, if the observer chooses to perform a ``which-way'' measurement, the photon traveled on one of the two paths A and B, while if she chooses to measure the interference, the photon traveled both ways. This suggests that the system had an initial state compatible with the observable, even if at that time the observable was not yet known.

\begin{figure}[ht]
\begin{minipage}[b]{0.48\linewidth}
\centering
\includegraphics[width=\textwidth]{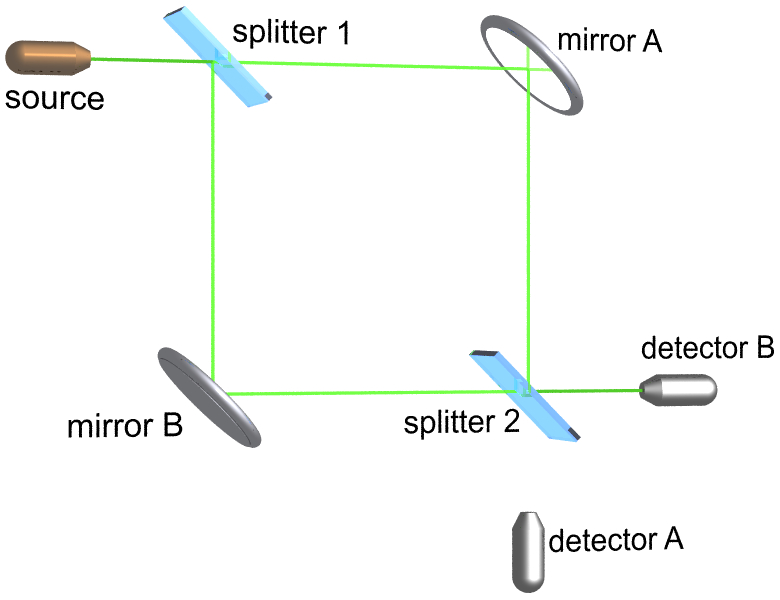}
{a) Both ways.\\
\vspace{4.5pt}}
\end{minipage}
\begin{minipage}[b]{0.48\linewidth}
\centering
\includegraphics[width=\textwidth]{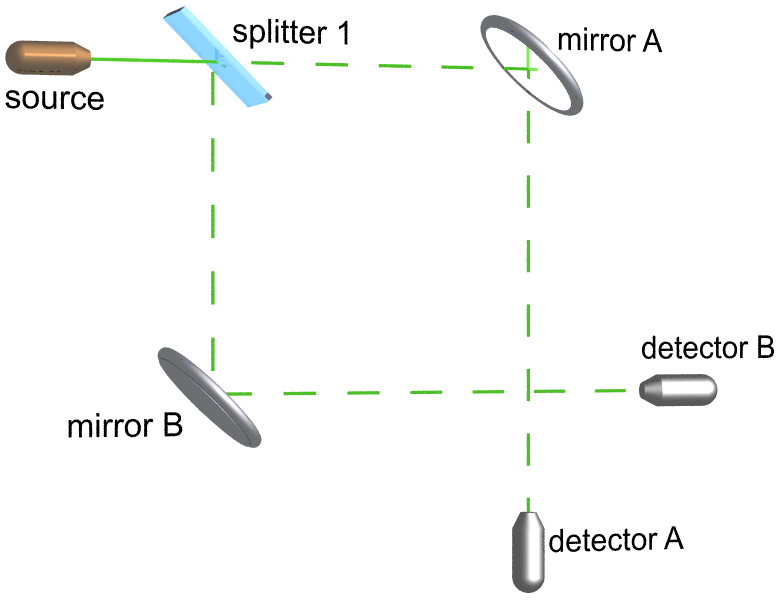}
{b) Which way.}
\end{minipage}
\caption{Delayed choice experiment with Mach-Zehnder's interferometer.}
\label{fig:mach_zehnder}
\end{figure}

A proposal to explain the apparently nonlocal EPR correlations by local and apparently retrocausal paths in spacetime, somewhat similar to that in section \sref{s:EPR}, was proposed by Olivier Costa de Beauregard in 1947 \cite{deBeauregard1953-DEBMQ}. Related ideas appear in \cite{Rietdijk1978retroactiveInfluence}.

Notable with respect to the way past seems to be influenced by future choices is the \textit{two-state vector formalism} (TSV) \cite{aharonov1964time,aharonov1988result,aharonov1991complete,aharonov2007TSV,aharonov2012future-past}. This way of describing things, by combining weak measurement with a description of quantum mechanics based on a state vector evolving towards the future, and another one towards the past, reveals something intimate about the nature of quantum mechanics. Also, an interesting interpretation of the wavefunction collapse using the TSV formalism was proposed in \cite{AharonovCohen2014MeasurementCollapse}. TSV formalism proved to be very proficient in thinking about new experiments, which explore the limits we thought quantum mechanics has. For example, in relation with the idea of apparent influence on the past from the future, in  \cite{aharonov2012future-past} is revealed how the EPR experiment can be modified to show that apparently future strong measurements affect the outcomes of weak measurements performed in the past. While TSV formalism provides a way to think about quantum systems in terms of future initial conditions, the conclusion of our analysis is different and independent of this approach.

A unitary interpretation of quantum mechanics, based on cellular automata, was proposed by 't Hooft \cite{hooft2011wave,Elze2014Action4CellularAutomata,tHooft2014CellularAutomatonInterpretationQM}. This interpretation also reveals the necessity of special initial conditions, which are ensured by \emph{superdeterminism}.

The many worlds interpretation \cite{Eve57,Eve73,dW71,dWEG73} is based solely on the {\schrod} unitary evolution. Because of this, is considered sometimes to provide a unitary account of the {\qmR} process. In this interpretation, there is a universal wavefunction which evolves unitarily, and which is decomposed as a result of a quantum measurement. Accordingly, to distinct outcomes correspond distinct branches, whose superposition is the universal wavefunction. But the unitary account for the measurement process presented here is different. In MWI, the unitary evolution is recovered only when all worlds are considered, but for each world or branch of the wavefunction, the {\qmR} process still appears to contradict the {\qmU} process. If in MWI unitary evolution is valid within a branch, then at the level of that branch, Property \ref{property_ic} holds.

\section{Interpretations of the property of restricted initial conditions}
\label{s_interp}

In this section I will explore some possible interpretations of the property of restricted initial conditions.

What does it mean that, for the evolution to remain unitary during measurement, the initial conditions have to be very special? Does this mean that the initial conditions ``guess'' the future choice of observable and the interaction with the measurement apparatus?

A possible interpretation is that the observer is ``predestined'' to choose the observable, so that the outcome is an eigenvalue of that observable. This explanation is called \textit{superdeterminism} (see \eg \cite{hooft2011wave}). Apparently, it denies the free will of the observer (we will not argue here if we should be concerned about the free will or not. The interested reader may consult \cite{CK06,Sto08f,Sto12QMa,Sto13bSpringer,aaronson2013ghostQM}).

Another possible view is that the initial conditions of the observed system are not decided until the observer chooses the measurement apparatus, and actually performs the experiment. In other words, the initial conditions themselves are delayed.

Which would be more acceptable, to admit that the observed system is ``predestined'' to become an eigenstate of the observable, or that the observer  is ``predestined'' to  choose an observable which is compatible with the observed system? Both interpretations are different from what one would expect causality to be like.

Another possibility is to consider that the initial conditions of the observed system and the measurement apparatus were already entangled, prior to the measurement, although this assumption also means that initial conditions are special (and Theorem \ref{thm_unitary_ic_density_matrix} shows that this cannot avoid anyway Property \ref{property_ic}). For example, in the case of the experiment with the Mach-Zehnder interferometer (fig. \ref{fig:mach_zehnder}), the measurement apparatus can be arranged in two ways, $\ket{\tn{observe both-ways}}$ and $\ket{\tn{observe which-way}}$, and the system can be found either in the $\ket{\tn{both-ways state}}$, or in one of the states $\ket{\tn{which-way state A}}$ and $\ket{\tn{which-way state B}}$ (depending whether the photon went through the arm labeled $A$, or the arm labeled $B$, in fig. \ref{fig:mach_zehnder}).
The density matrix of the total system is therefore a mixture made of the states
\begin{equation}
\label{eq_superposition_observables}
\begin{array}{l}
\ket{\tn{observe both-ways}}\ket{\tn{both-ways state}},\\
\ket{\tn{observe which-way A}} \ket{\tn{which-way state A}}, \tn{ and}\\
\ket{\tn{observe which-way B}} \ket{\tn{which-way state B}}.\\
\end{array}
\end{equation}
We can consider this a statistical ensemble, and eventually find that only one of the states in the mixture is realized, depending on the choice of the observable. This reformulation in terms of entanglement may seem	 more reasonable, but in fact such a statistical ensemble will just be a statistical ensemble of states having very special initial conditions. Only a subset of measure zero of $\hilbert_{\mu}\otimes\hilbert_{\psi}$ is allowed for the initial conditions. Hence, the problem remains the same: the initial conditions have to be very special in a way which seems to anticipate the future choices. We cannot avoid Property \ref{property_ic}, of restricted initial conditions.

If the evolution is indeed always unitary, and the {\qmR} process is just a special case of the {\qmU} process, then the author's preferred explanation of the restricted initial conditions comes from the \textit{block universe} of relativity theory \cite{Sto12QMc,Sto13bSpringer}. At the core of quantum mechanics there is a picture in which the Hamiltonian governs the time evolution. On the other hand, special and general relativity present the universe as a four-dimensional block. The time evolution can be obtained by foliating the space-time manifold in space+time \cite{adm2008admRepublication}. The role of time in quantum mechanics seems to conflict with that in the theory of relativity. However, if we think of the solutions to {\schrod}'s equations from the viewpoint of the block universe, we have to impose a \textit{global consistency principle}, stating that, among the possible solutions, one has to keep only those which are globally self-consistent.

In general, when the evolution of a system is described by a set of partial differential equations (PDE), the natural thing to do is to study the initial value problem. The initial value problems have two parts, the PDE, and the initial conditions. But it is not always guaranteed that there is always a global solution satisfying a particular set of initial conditions. It may be the case that a solution satisfying particular initial conditions does not exist globally. In this case, global consistency eliminates the initial conditions leading to inconsistencies, leaving us only with a restricted set of initial conditions.

Let us see an example. A three-dimensional version of the global consistency principle can be anticipated in {\schrod}'s work of deriving the discrete energy spectrum of the electron in the atom from boundary conditions on the sphere at infinity \cite{Sch26}. His insight was to require that whatever the electron's wavefunction may be around the atom, to be physically admissible, it has to extend to infinity in a self-consistent manner. This led straightforwardly to solutions representing electrons as standing waves around the atom as proposed by de Broglie \cite{dB24}, and explained the spectrum of the Hydrogen atom.

Similarly, one should expect something like this in four dimensions, that is, in spacetime. Global consistency principle implies in a straightforward manner that one should rule out the solutions whose initial conditions will be invalidated by the future experimental setups. 
For example, in the EPR experiment Alice and Bob should obtain consistent outcomes for their measurements. Local solutions to the {\schrod} equation exist in Alice's and Bob's laboratory, but when we patch local solutions together, the results have to be consistent. It is not allowed for example for both Alice and Bob to obtain the spin up along the same axis, this would be inconsistent with the {\schrod} equation itself, and with the initial condition that the two particles were forming a singlet state \cite{Sto12QMc,Sto13bSpringer}\footnote{Patching local solutions together to obtain global solution, and the obstructions preventing this, are studied in \emph{sheaf theory} \cite{bredon1997sheaf}.}. The apparently non-local correlations which make quantum mechanics seem so strange can be viewed as a consequence of the global consistency principle.

These ideas are in agreement with the experiment presented in \cite{aharonov2012future-past}, leading the authors of that article to conclude that

\begin{quote}
what appears to be nonlocal in \emph{space} turns out to be perfectly local in \emph{spacetime}.
\end{quote}

Only the future will tell if global consistency is able to solve the measurement problem in a unitary way.

\subsection*{Acknowledgements}
The author cordially thanks 
Eliahu Cohen, Hans-Thomas Elze, Florin Moldoveanu, and the referees,
for very helpful comments and suggestions.

\end{document}